\documentclass[runningheads]{llncs}

\usepackage[T1]{fontenc}

\usepackage{hyperref}
\usepackage{color}

\usepackage{amsmath}
\usepackage{amssymb}
\usepackage{xcolor}

\newcommand{\R}{{\mathbb R}}
\newcommand{\Rnn}{{\mathbb R}_{\ge 0}}
\newcommand{\B}{\{0, 1\}}

\newcommand{\gt}{\mathrm{GT}}
\newcommand{\gte}{\mathrm{GTE}}
\newcommand{\lt}{\mathrm{LT}}
\newcommand{\lte}{\mathrm{LTE}}
\newcommand{\eq}{\mathrm{EQ}}
\newcommand{\neqg}{\mathrm{NEQ}}
\newcommand{\id}{\mathrm{id}}
\newcommand{\swap}{\mathrm{SWAP}}
\newcommand{\nw}{\mathrm{NW}}
\newcommand{\snw}{\mathrm{SNW}}
\newcommand{\perm}{\mathrm{PERM}}
\newcommand{\ce}{\mathrm{CE}}
\newcommand{\odd}{\mathrm{ODD}}
\newcommand{\even}{\mathrm{EVEN}}
\newcommand{\hamm}{\mathrm{H}}
\newcommand{\tresh}{\mathrm{T}}
\newcommand{\match}{\mathrm{MATCH}}

\begin{document}

\title{Succinct QUBO formulations for permutation problems by sorting networks}

\titlerunning{Succinct QUBOs for permutations}

\author{
Katalin Friedl\inst{1} \and
Levente Gegő\inst{1} \and
László Kabódi\inst{1} \and
Viktória Nemkin\inst{1}}

\authorrunning{K. Friedl et al.}

\institute{Department of Computer Science and Information Theory\\Budapest University of Technology and Economics\\
\email{\{friedl, lgego, kabodil, nemkin\}@cs.bme.hu}}

\maketitle

\begin{abstract}
Quadratic Unconstrained Binary Optimization (QUBO) is a standard NP-hard optimization problem. Recently, it has gained renewed interest through quantum computing, as QUBOs directly reduce to the Ising model, on which quantum annealing devices are based.

We introduce a QUBO formulation for permutations using compare-exchange networks, with only $O(n \log^2 n)$ binary variables. This is a substantial improvement over the standard permutation matrix encoding, which requires $n^2$ variables and has a much denser interaction graph. A central feature of our approach is uniformity: each permutation corresponds to a unique variable assignment, enabling unbiased sampling.

Our construction also allows additional constraints, including fixed points and parity. Moreover, it provides a representation of permutations that supports the operations multiplication and inversion, and also makes it possible to check the order of a permutation.

This can be used to uniformly generate permutations of a given order or, for example, permutations that commute with a specified permutation.

To our knowledge, this is the first result linking oblivious compare-exchange networks with QUBO encodings. While similar functionality can be achieved using permutation matrices, our method yields QUBOs that are both smaller and sparser. We expect this method to be practically useful in areas where unbiased sampling of constrained permutations is important, including cryptography and combinatorial design.

\end{abstract}

\section{Introduction}
\label{sec-intro}

Permutations are a fundamental concept in computer science, a wide range of important problems can be viewed as finding a permutation that satisfies a prescribed set of constraints, including various tasks in logistics, planning, and scheduling. Generating random permutations is a basic building block of randomized algorithms and also comes up in probability theory, in card shuffling problems. Its complexity was investigated in different environments: in the 1960s a linear time algorithm to uniformly sample permutations was given in~\cite{durstenfeld-permu-gen}. Later, the parallel version became interesting. Here, shuffling networks of logarithmic depth have been considered to generate almost random permutations~\cite{czumaj-permu-switching} and the problem has also been studied in details in the RAM model~\cite{torben-permu-ram} as well as recently in the cell-probe model~\cite{alekseev-permu-cell} and quantum circuits~\cite{bibhas-permu-quantum}. Furthermore, there is another line of research on the mixing time of different shuffling methods. The simple Thorp shuffle is good enough to produce nearly
uniform permutations~\cite{ben-thorp},~\cite{ben-thorp-improved}.

Quadratic unconstrained binary optimization (QUBO) has been dubbed as a unified framework for modeling and solving combinatorial optimization problems, because many such problems can be naturally described in this form see, for example~\cite{kochenberger-glover-unified-2004},~\cite{lucas-qubo-np-2014},~\cite{punnen-qubo-2022} for an introduction to this area. Formally, QUBO is defined as minimizing a quadratic objective function over the Boolean hypercube,
\[
\min_{x \in \B^n} \sum_{i,j} q_{i,j} x_i x_j.
\]

In the context of QUBO, permutations are typically encoded using $\Theta(n^2)$ qubits and $\Theta(n)$ interactions per-qubit, via one-hot encoding, also known as permutation matrix encoding~\cite{mayowa-permu-weights},~\cite{gbgl-permu-tech}. Recently, domain-wall encoding has been proposed in~\cite{codognet-domain-wall} as an alternative, which halves the numer of per-qubits interactions. An easy to see lower bound on the number of qubits necessary to encode permutations is $\Omega(n \log n)$. This has been shown to be achieveable in a different setting, with higher-order terms allowed (HOBO) by \cite{cimbi-hobo} using the factorial number system.

Instead of nearly uniformly, our goal is to generate permutations uniformly, which we achieved by using compare-exchange networks. In this paper, we introduce a QUBO formulation for permutations using compare-exchange networks with only $O(n \log^2 n)$ binary variables and $O(\log n)$ interactions per-qubit. Our main contribution is the general idea to connect oblivious algorithms to QUBO (and in general, mathematical programming) formulations: since these algorithms follow a fixed sequence of operations regardless of the input, their correct behavior can be described by a mathematical formula. The main idea is to describe the relation consisting of the correct (input, output) pairs of the network as a quadratic optimization problem. Beyond that, the formulations themselves are easily obtained. We show (Theorem~\ref{thm-network}) that this can be done with $O(k(n+m))$ variables, where the network has $n$ input numbers of $k$ bits each and $m$ is the number of gates in the network. 

We use this framework to describe some constraints about permutations, including parity, order, elements of a conjugacy class, involutions, derangements, and fixed points. Constraints like these are relevant in fields such as cryptography, in S-box design~\cite{crypto-aes}, or combinatorial designs, where the Latin Square Completion problem and its variants serve as a foundation for applications like sports tournament scheduling~\cite{nemhauser-sport-scheduling}. Moverover, we address the problem of Permutation Pattern Matching~\cite{bose-ppm} via stable sorting networks. 

With the assumption that a (classical or quantum) solver samples the solutions uniformly, our construction provides a way to uniformly sample permutations in general, or with the previous constraints.

A notable limitation of our framework is its difficulty in modeling graph-theoretic problems such as the Traveling Salesperson Problem (TSP): in these cases, representing vertices as binary strings makes edge verification a non-trivial task.

\section{Preliminaries}
\label{sec-prelim}
We recall basic notions for formulating QUBO problems. It is well known that every Boolean function $f : \B^k \to \B$ admits a unique multilinear polynomial $g : \R^k \to \R$ such that $f(x)=g(x)$ for all $x \in \B^k$. However, since $g$ is generally not quadratic, we consider a relaxed representation, that will not be unique. Here, the satisfying assignments of $f$ correspond to the roots of a non-negative multilinear polynomial $g$ with additional variables.

\begin{definition}[Polynomial representation]
\label{def-represent}
Let $f : \B^k \to \B$ be a Boolean function. The non-negative multilinear polynomial $g: \R^{k+m} \to \Rnn$ \emph{represents} $f$ when
\begin{equation*}
\min_{z \in \B^m} g(x,z) ~
\begin{cases}
~=0 & \text{ if }~ f(x) = 1 \\
~>0 & \text{ if }~ f(x) = 0
\end{cases}
\end{equation*}
holds for all $x \in \B^k$. The $z$ are called the auxiliary variables.

A relation is represented by the polynomial $g$ if $g$ represents the indicator function of the relation.
\end{definition}

From now on, we refer to functions like $g$ simply as \emph{polynomials} and by Boolean function or relation, we usually mean a polynomial representing it.

\begin{definition}[Uniformity]
A polynomial $g(x,z)$ is \emph{uniform} if, for every $x$ satisfying $f(x)=1$, the number of $z$ assignments where $g(x,z)=0$ is independent of $x$.
\end{definition}

In the context of this paper, this number is typically $1$. With the assumption that a (classical or quantum) solver finds the roots of $g$ uniformly, this provides a way to uniformly sample satisfying assignments of $f$.

\subsection{Quadratization}

It is known that every Boolean function can be represented by a quadratic polynomial. Such a polynomial can be obtained by reducing any higher-degree representation to a quadratic form, while preserving the underlying Boolean function. This reduction can be achieved, for example, by introducing auxiliary variables using the following observations. See, for example~\cite{boros-quadrat}, or for a more extensive review~\cite{dattani-quadrat}.

A higher degree negative monomial $-\prod_{i=1}^t x_i$, where $x_i \in \B$ can be replaced by introducing one auxiliary variable $z \in \B$ with the quadratic term 
\begin{equation}
\label{eq-quad-neg}
z(2t-1-2\sum_{i=1}^t x_i)
\end{equation}
because its minimum for fixed $x_i$ equals to $-\prod_{i=1}^t x_i$. Since this minimum is taken uniquely when $z=\prod_{i=1}^t x_i$, this transformation preserves uniformity.

The degree of a higher degree positive monomial can be decreased by replacing a subproduct $x_i x_j$ with a new variable $z$.
To guarantee that $z=x_i x_j$, a quadratic penalty term
\begin{equation}
\label{eq-quad-pos}
\alpha(x_i x_j - z(2 x_i + 2 x_j -3))
\end{equation} 
is added to the polynomial. This term is $0$ if $z = x_i x_j$ holds and is at least $1$ otherwise.

The penalty weight $\alpha$ is set to be strictly larger than the coefficient of the original monomial, so that the incorrect assignment of $z$ results in a strictly larger function value. It follows, that this transformation preserves uniformity.

Eq. (\ref{eq-quad-pos}) can also be applied to negative monomials, in this case $\alpha$ is set to be larger than the absolute value of the coefficient.

Notice that the resulting $g$ is not an exact quadratization of $f$. When $f(x)>0$, the minimum of $g(x,z)$ over $z$ need not equal to the value of $f(x)$. For tasks where the goal is to find a zero assignment of $f$, such a quadratization is sufficient.

\subsection{Basic QUBOs}

For completeness, we list here uniform quadratic polynomial representations of some basic Boolean functions. Note that, for simplicity, the auxiliary variables can be omitted from the parameter lists.

Let $x \in \B^n$ and $k \in \{0, \dots, n\}$. The polynomial for the Hamming-weight function is
\begin{equation}
\label{eq-hamming}
\hamm_k(x) = \left(\sum_{i=1}^n x_i-k\right)^2,
\end{equation}
which is $0$ when the sum of $x$ is $k$ and positive otherwise.

Let $t=\lfloor \log(n+1) \rfloor$ and auxiliary variables $y \in \B^t$. Define the auxiliary integer variable 
\begin{equation*}
Y = \sum_{i=1}^t 2^{i-1} y_i,
\end{equation*}
its range includes $\{0, \dots, n\}$. The polynomial for the treshold function is
\begin{equation}
\label{eq-treshold}
\tresh_k(x) = \left(\sum_{i=1}^n x_i -k-Y\right)^2,
\end{equation}
for which the minimum over $Y$ is $0$ when the sum of $x$ is at least $k$ and positive otherwise.

The polynomials for the even and odd functions are
\begin{align}
\label{eq-odd}
\even(x) =& \left(\sum_{i=1}^n x_i - (Y-y_1) \right)^2\\
\label{eq-even}
\odd(x)  =& \left(\sum_{i=1}^n x_i - (Y-y_1+1) \right)^2.
\end{align}
The minimum of these polynomials over $Y$ is $0$ when the parity of the sum of $x$ is correct and positive otherwise.

Moreover, since each value of $Y$ corresponds to a unique assignment of the bits $y$, all of the polynomials above are uniform representations using $O(\log n)$ auxiliary ariables.

\subsection{Permutation matrix encoding}
The well-known approach to encode a permutation $\pi$ of $\{1, \dots, n\}$ is to introduce variables $p_{i,j}\in\B$ for $i,j \in \{1, \dots, n\}$, such that $p_{i,j}=1 \Leftrightarrow \pi(i)=j$. These must form a permutation matrix, which is enforced by the constraints
\begin{equation*}
\sum_{j=1}^n p_{i,j} = 1 \quad (1 \leq i \leq n) \quad \text{and} \quad
\sum_{i=1}^n p_{i,j} = 1 \quad (1 \leq j \leq n).
\end{equation*}
These can be described by (\ref{eq-hamming}) as
\begin{equation*}
\left(\sum_{j=1}^n p_{i,j} - 1\right)^2 \quad (1 \leq i \leq n) \quad \text{and} \quad
\left(\sum_{i=1}^n p_{i,j} = 1\right)^2 \quad (1 \leq j \leq n).
\end{equation*}

This uses exactly $n^2$ binary variables. Moreover, each variable $p_{i,j}$ interacts with $O(n)$ other variables: after expanding the expressions above, it appears in $O(n)$ quadratic terms, paired with the remaining variables of row $i$ and coumn $j$.

\section{Sorting network formulation}

Compare-exchange networks \cite{knuth-1998-ce}, also known as sorting networks, are a well-known algorithm family that can sort $n$ numbers in an oblivious way, i.e., the sequence of operations does not depend on the actual input, only on its size. Such a network consists of $n$ parallel lines that initially carry the input numbers. Gates act on pairs of lines and can modify the values they carry. In our formulation, variables represent the input and output of a gate with some additional control variables.

A sorting network uses \textit{compare-exchange gate}s. This gate compares the two numbers on its inputs and outputs them in the correct order. These gates can be decomposed into a \textit{greater than gate} and a \textit{controlled swap gate}.

A \textit{greater than gate} takes two input numbers $(x, y)$ and outputs a single comparison bit $c$, where $c = 1$ if $x > y$, and $c = 0$ otherwise.

A \textit{controlled swap gate} takes two input numbers $(x, y)$ and an input bit $c$ for control. If $c = 0$, the outputs are $(x, y)$ unchanged; if $c = 1$, the outputs are swapped $(y, x)$.

A \textit{compare-exchange gate} combines these: it takes two input numbers $(x, y)$, compares them, and if $x > y$, it swaps them; otherwise, it leaves them unchanged. Thus its output is always the pair $(x, y)$ sorted in non-decreasing order.

The correct behavior of a gate can be described by a relation on its variables. For this relation, a polynomial representation is given that is quadratized. Later, these \emph{gate-polynomials} are combined to represent the entire network. Their value will be zero if and only if all gates ``behave properly'', each gate-polynomial representing the relation of its gate.

\subsection{Polynomial representation for controlled swap gates}

Let $x_1, x_2$ correspond to the input numbers of the controlled swap gate, $y_1, y_2$ to the output numbers and $c$ to the control bit.

\begin{definition}[Swap relation] \label{def-swap}
Assume the variables $x_1, x_2, y_1, y_2$ can hold elements from $\B^k$ and $c$ from $\B$. The $\swap$ relation consists of 
the tuples $(x_1, x_2, y_1, y_2, c)$ where $x_1=y_1, x_2=y_2, c=0$ or $x_2=y_1, x_1=y_2, c=1$ holds.
\end{definition}

\begin{proposition}[Swap polynomial] \label{prop-swap}
There is a uniform quadratic polynomial representation of the $\swap$ relation, $\swap(x_1, x_2, y_1, y_2, c)$ on these $4k+1$ variables that uses $2k$ auxiliary variables.
\end{proposition}
\begin{proof}
First consider the bit case, i.e., when $k=1$.
The quadratic polynomial $g_s(x_1, x_2, c) = (1-c) x_1 + c x_2$ selects one of $x_1$ or $x_2$, depending on the value of $c$.
Using this, the relation of the controlled swap for $k=1$ can be represented by the polynomial
\begin{equation*}
(y_1 - g_s(x_1, x_2, c))^2 + (y_2 - g_s(x_2, x_1, c))^2.
\end{equation*}
To make it 0, both terms have to be 0, so depending on the value of $c$, either $y_1=x_1$ and $y_2=x_2$, or  $y_1=x_2$ and $y_2=x_1$ must hold.
This can be quadratized, for example, by introducing auxiliary variables for $c x_1$ and $c x_2$ as in (\ref{eq-quad-pos}). With these substitutions, the polynomial $g_s$ becomes linear, and the resulting polynomial, together with the added penalty terms called $\swap_1$ is quadratic.

When $k > 1$, the controlled swap can be performed as $k$ controlled swaps on the individual bits, each of them using the same control bit.
Therefore, a quadratization $\swap(x_1, x_2, y_1, y_2, c)$ for the general controlled swap can be obtained by summing up $\swap_1$ applied to the 1st, 2nd, $...$, $k$th bit of the parameters $x_1, x_2, y_1, y_2$ using the same control bit $c$ in each case.

The 1-bit case uses 2 auxiliary variables, so the $k$-bit version uses $2k$ of them. Since (\ref{eq-quad-pos}) preserves uniformity, this representation is uniform.
\end{proof}

\subsection{Polynomial representation for comparison gates}

Let $x, y$ correspond to the input numbers of the greater than gate and $c$ to the output comparison bit. Here $c=1$ means that $x>y$ and $c=0$ means that $x \leq y$.

\begin{definition}[GT relation] \label{def-gt}
Assume that the variables $x, y$ can hold elements from $\B^k$ and $c$ from $\B$. The $\gt$ relation consists of the triplets $(x, y, c)$ where $x \le y, ~ c=0$, or $x > y, ~ c=1$ holds. 
\end{definition}

For bits $x, y, c \in \B$, consider the polynomial
\begin{equation}
\label{eq-comp-h}
h(x, y, c) = 1+(1-2c) (x - y).
\end{equation}
Notice that $h$ is non-negative and can only be $0$ when $x \neq y$, more specifically when $x<y$ and $c=0$ or $x>y$ and $c=1$. This is used in the formulation below.

\begin{proposition}[GT polynomial] \label{prop-gt}
There is a uniform quadratic polynomial representation of the $\gt$ relation, $\gt(x, y, c)$ on these $2k+1$ variables that uses $3k+1$ auxiliary variables.
\end{proposition}

\begin{proof}
Let the bits of the two parameters be $x = x^k  \ldots x^2 x^1$ and  $y = y^k  \ldots y^2 y^1$, where $x^k$ and $y^k$ are the most significant bits. To compare these numbers, we first look for the highest differing bit, if it exists. To do this, we introduce $k+1$ auxiliary variables $\{p_0, p_1, \dots, p_k\}$ each $p_i \in \B$.

If the numbers $x$ and $y$ are the same, let $p_0=1$ and $c=0$. Otherwise if the highest differing bit is the $i$th one let $p_i=1$, and when $x^i < y^i$ let $c=0$ while when $x^i > y^i$ let $c=1$. This is described by the following conditions.

\smallskip
(i) Exactly one of the $p_i$ is $1$, by the penalty
\begin{equation*}
H_1(p) = \left(\sum_{i=0}^{k} p_i-1 \right)^2
\end{equation*}
from (\ref{eq-hamming}).

\smallskip
(ii) If $p_i=1$ then we need that $x^k = y^k, \ldots, x^{i+1} = y^{i+1}$. Defining the function
\begin{equation*}
g_i(x,y)=\sum_{j=i+1}^{k} x^j + y^j - 2 x^j y^j,
\end{equation*}
which counts the number of differing bits between the $k$th and the $(i+1)$th, the penalty is
\begin{equation*}
 \sum_{i=0}^k p_i g_i(x, y).
\end{equation*}

\smallskip
(iii) If $p_i=1$ and $0<i$, then a comparison of the $i$th bits is needed, by the penalty
\begin{equation*}
\sum_{i=1}^{k} p_i h(x^i, y^i, c),
\end{equation*}
using the $h$ from (\ref{eq-comp-h}).

\smallskip
(iv) If $x=y$, then $c=0$ must hold, by the penalty $cp_0$.

\smallskip
Putting (i)-(iv) together we get
\begin{equation}
H_1(p) + \sum_{i=0}^k p_i g_i(x, y) + \sum_{i=1}^{k} p_i h(x^i, y^i, c) + c p_0.
\end{equation}

The sums from (ii) and (iii) contain cubic polynomials, they can be quadratized by using further $2k$ auxiliary variables, for example introducing variables for $x^i y^i$ and $p_i c$ as in (\ref{eq-quad-pos}) and adding the appropriate penalty functions to obtain the quadratic polynomial $\gt(x,y,c)$.

This representation uses $k+1+2k=3k+1$ auxiliary variables. Since the values of the $p_i$ are uniquely determined by the values of $x$ and $y$ and (\ref{eq-quad-pos}) preserves uniformity, this representation is uniform.
\end{proof}

Other types of comparison gates can be similarly constructed. The greater than or equal gate $\gte(x,y,c)$ is obtained by replacing the $c p_0$ term in condition (iv) with $(1-c)p_0$ to guarantee $c=1$ when $x=y$.

The gates for less than ($\lt$) and less than or equal ($\lte$) follow immediately by $\lt(x,y,c) = \gt(y,x,c)$ and $\lte(x,y,c) = \gte(y,x,c)$.

Finally, the gates equal ($\eq$) and not equal ($\neqg$) are constructed by removing conditions (iii) and (iv) and replacing the variable $p_0$ with $c$ for $\eq$ and $1-c$ for $\neqg$. Since condition (iii) is not present, these use $k$ auxiliary variables less than the original $\gt$ gate.

\begin{corollary}[Comparison gates]
\label{cor-comp}
There are uniform quadratic polynomial representations for

\smallskip
(i) the relations $\lt(x,y,c)$, $\lte(x,y,c)$, $\gt(x,y,c)$ and $\gte(x,y,c)$ using $3k+1$ auxiliary variables and

\smallskip
(ii) the relations $\eq(x,y,c)$ and $\neqg(x,y,c)$ using $2k+1$ auxiliary variables.
\end{corollary}

\subsection{Polynomial representation for CE networks}
\label{network}

First we give a quadratic gate-polynomial for the compare-exchange gates then show that these can be combined to describe the correct behavior of the entire network.

A compare-exchange gate is simply a combination of a greater than gate and a controlled swap gate, where the control bit depends on the result of the comparison. Let $x_1, x_2$ correspond to the input numbers of the CE gate, $y_1, y_2$ to the output numbers and $c$ to the control bit.

\begin{definition}[CE relation] \label{def-ce}
Assume that $x_1, x_2, y_1, y_2$ can hold elements from $\B^k$ and $c$ from $\B$. The $\ce$ relation consists of the tuples $(x_1, x_2, y_1, y_2, c)$ where $x_1 \leq x_2$, $c=0$, $x_1=y_1$, $x_2=y_2$ or $x_1 > x_2$, $c=1$, $x_1=y_2$, $x_2=y_1$ holds.
\end{definition}

\begin{proposition}[CE polynomial] \label{prop-ce}
There is a uniform quadratic polynomial representation of the $\ce$ relation, $\ce(x_1, x_2, y_1, y_2, c)$ on these $4k+1$ variables that uses $5k+1$ auxiliary variables.
\end{proposition}

\begin{proof}
The quadratic polynomial can be obtained from the previous $\swap$ and $\gt$ polynomials. The control of $\swap$ is the same variable as the comparison of $\gt$.
\begin{equation}
\ce(x_1, x_2, y_1, y_2, c) = \swap(x_1, x_2, y_1, y_2, c) + \gt(x_1, x_2, c)
\end{equation}
By Propositions~\ref{prop-swap} and \ref{prop-gt}, this uses $2k+3k+1$ auxiliary variables total, and is a uniform representation.
\end{proof}

Consider now a compare-exchange network on $n$ input numbers, each of them having $k$ bits. The network has $n$ lines that can hold these numbers. At the beginning of the network, the lines hold the input numbers, and at the end of the network, the outputs are the same numbers in increasing order.

Assume that the network has $m$ gates. Although a network is typically used as a parallel algorithm, let us consider a sequential version of this algorithm, and number the gates from 1 to $m$ accordingly, so gate $i$ can use the output of gate $j$ only when $j < i$. Let $(\ell_{i,1}, \ell_{i,2})$ be the pair of lines of the $i$th gate ($1 \le \ell_{i,1} < \ell_{i, 2} \le n$).

Let $x = (x_1, \dots, x_n)$, each $x_i \in \B^k$ be the input numbers and $y = (y_1, \dots, y_n)$, each $y_i \in \B^k$ be the output numbers of the sorting network. Let $c \in \B^m$ be the control bits of the CE gates.

\begin{definition}[Sorting network relation] \label{def-network}
Assume the variables $x, y$ can hold elements from $\B^{nk}$ and $c$ from $\B^m$. The $\nw$ relation consists of the triplets $(x,y,c)$ that represent a correctly behaving sorting network, i.e. the $y$ numbers are the $x$ numbers in ascending order and the $c$ correspond to the control bits of the gates in the network.
\end{definition}

\begin{theorem}[Sorting network polynomial] \label{thm-network}
There is a uniform quadratic polynomial representation of the sorting network relation, $\nw(x, y, c)$ on these $2nk+m$ variables that uses $m(7k+1)-nk$ auxiliary variables.
\end{theorem}
\begin{proof}

An input variable of a CE gate is either an input variable of the entire network, or the output variable of a preceding gate. Each CE gate introduces the following new variables: one control bit, $2k$ output variables, and $5k+1$ auxiliary variables, as described in Proposition ~\ref{prop-ce}. Altogether, this amounts to $m(7k+2)$ variables, which includes the control bits and the output variables of the network. Assuming each line has at least one gate, the number of auxiliary variables is $m(7k+1)-nk$. 

Summing up the $\ce$ polynomials we obtain a quadratic polynomial that is non-negative (since each $\ce$ is such) and $0$ if and only if all the individual polynomials are $0$, which means that the computation is correctly described at each gate.

Formally, let variable $z_{j,0}=x_j$ denote the input variable of line $j$. After each gate on the line, the output variable of that gate is denoted by $z_{j,p}$, where $p$ is the position of the gate along the line, starting from index $1$.

For gate $i$, when its position is $p_{i,1}$ on line $\ell_{i,1}$ and $p_{i,2}$ on line $\ell_{i,2}$ its input variables are $z_{\ell_{i,1},p_{i,1}-1}$ and $z_{\ell_{i,2}, p_{i,2}-1}$, while its output variables are $z_{\ell_{i,1}, p_{i,1}}$ and $z_{\ell_{i,2}, p_{i,2}}$. The last variables on each line correspond to the output variables of the network. The quadratic polynomial is then
\begin{equation*}
\nw(x,y,c) = \sum_{i=1}^{m} \ce(z_{\ell_{i,1}, p_{i, 1}-1}, z_{\ell_{i,2}, p_{i, 2}-1}, z_{\ell_{i,1}, p_{i, 1}}, z_{\ell_{i,2}, p_{i, 2}}, c_i).
\end{equation*}
Furthermore, by Proposition ~\ref{prop-ce}, this representation is uniform.
\end{proof}

\subsection{Polynomial representation for permutations}
\label{sec-perm}

Application of Theorem \ref{thm-network} yields the following uniform quadratic polynomial representation for permutations with only $O(n \log^2 n)$ variables.

\begin{definition}[Permutation relation] \label{def-perm}
Consider a permutation of the numbers $\{1, \dots, n\}$ and let $k=\lfloor\log(n+1)\rfloor$ be their bit length. Asumme the variables $x$ can hold elements from $\B^{nk}$. The $\perm$ relation consists of $x$ that are permutations.
\end{definition}

Furthermore, let $\id=(1, \dots, n)$ denote the identity permutation, where $\id_i = i$ is a constant binary number on $k$ bits.

\begin{corollary}[Permutation polynomial] \label{cor-perm}
There is a uniform quadratic polynomial representation of the relation $\perm$, $\perm(x)$ on these $O(n \log n)$ variables using $O(n \log^2 n)$ auxiliary variables, where the number of interactions is $O(\log n)$ per variable.
\end{corollary}

\begin{proof}
Take a sorting network with $m=O(n \log n)$ gates, such as AKS~\cite{aks-sorting-1983}. Use the uniform quadratic polynomial representation of this network $\nw(x,y)$ from Theorem~\ref{thm-network} with the restriction that the output variables of the network, $y$ are replaced with the constans $\id$.
\begin{equation*}
\perm(x) = \nw(x, \id)
\end{equation*}
Here, the control bits of the network are considered auxiliary and omitted from the parameter list.

The total number of variables is $O(m k)$, which is $O(n \log^2 n)$. In a sorting network the path of each input number is uniquely determined by the permutation, therefore the correct setting of the control bits is unique. It follows, that this representation is uniform.

Since each variable in the formulation is involved in at most two $\ce$ gates, which are represented using $O(\log n)$ variables, it follows that the number of interactions each variable is involved in is $O(\log n)$ as well, meaning that this formulation is sparse.
\end{proof}

\begin{remark}
Instead of AKS, any sorting network can be used. If it has $m$ gates, then the number of variables is $O(m \log n)$.
\end{remark}

Recall that a sorting algorithm is \emph{stable} if it preserves the relative order of equal elements. A stable sorting network can be obtained by numbering the lines from $1$ to $n$ and appending the corresponding line number at the end of each input variable. Then, whenever two original inputs are equal, comparisons are decided using the appended line numbers, ensuring that their relative order remains intact.

Formally, let $(x, \id)$ denote the sequence
\begin{equation*}
((x_1, 1), \dots, (x_n, n)),
\end{equation*}
where the $i$th input is $(x_i, i)$, which has $2k$ bits, the first $k$ coming from original $i$th input number $x_i$ and the second $k$ coming from the constant $i$.

More generally, for $x,y \in \B^{nk}$, i.e. $n$ variables of $k$ bits each, a \emph{composite input} $(x, y)$ denotes the sequence
\begin{equation*}
((x_1, y_1), \dots, (x_n, y_n)),
\end{equation*}
where the $i$th input number is $(x_i, y_i)$, formed by concatenating the variables $x_i, y_i$.

\begin{corollary}[Stable sorting network]
A stable sorting network $\snw(x, y, c)$ is obtained with
\begin{equation*}
\snw(x,y,c) = \nw((x, \id), (y, z), c),
\end{equation*}
where the $z \in \B^{nk}$ are auxiliary variables.
\end{corollary}

\section{Constraints on permutations}

There are several restrictions on permutations that can be easily formulated, either by specifying the values of some variables or adding further penalties. In this section, we list some of these.

Let $\pi,\pi',\sigma,\tau$ be permutations of the numbers $\{1, \dots, n\}$ represented by $x,x',y,z \in \B^{nk}$ and let variables $c \in \B^m$.

\begin{theorem} \label{thm-const-struct}
The following constraints can be enforced by adding further penalties to $\perm(x)$, while the number of auxiliary variables remains $O(n \log^2 n)$.

(a) $\pi(i) = j$

(b) $\pi(i) \neq j$

(c) $i$ is a fixed point of $\pi$, i.e. $\pi(i) = i$

(d) $\pi$ is a derangement, i.e. has no fixed points

(e) $\pi \neq \tau$ for a fixed permutation $\tau$
\end{theorem}

In the permutation matrix representation these can be achieved while the number of auxiliary variables remains $O(n^2)$. The proof of the theorem and the more detailed comparison with the permutation matrix representation can be found in Appendix~\ref{appx-const-struct}.

\begin{corollary}
Permutations with constraints such as (a)-(e) can be uniformly sampled with our construction.
\end{corollary}

For further constraints we first show how a sorting network can be modified to carry additional information alongside the values being sorted. Each input is a $k$-bit string, that is viewed as a concatenation of a $b$-bit \emph{key} and a $(k-b)$-bit \emph{payload}. At each gate, the comparison is performed using only the $b$ bits of the key, while the swap is performed on the entire $k$-bit string. 

Using composite inputs, we denote the separation between the key and the payload with a semi-colon: $(x; y)$ means, that the $i$th input is $(x_i, y_i)$, where $x_i$ is the key and $y_i$ is the payload.

\begin{theorem} \label{thm-const-gt}
The following constraints can be enforced by adding further penalties to $\perm(x)$, while the number of auxiliary variables remains $O(n \log^2 n)$.

(a) $\sigma = \pi \pi'$

(b) $\pi$ is an involution, i.e., $\pi^2 = I$

(c) $\pi$ and $\pi'$ commute, i.e., $\pi' \pi = \pi \pi'$

(d) $\pi'$ is a conjugate of $\pi$, i.e. $\pi' = \sigma^{-1} \pi \sigma $

(e) the parity of $\pi$ is even or odd

(f) $\pi^r = I$

(g) the order of $\pi$ is $r$
\end{theorem}

In the permutation matrix representation these can be achieved while the number of auxiliary variables remains $O(n^2)$, $O(n^3)$ or $O(n^4)$, depending on the constraint. The proof of the theorem and the more detailed comparison with the permutation matrix representation can be found in Appendix~\ref{appx-const-gt}.

\begin{corollary}
Permutations that are involutions, permutations commuting with a fixed $\pi$, permutations conjugate to a fixed $\pi$, even and odd permutations, permutations with a given order can be uniformly sampled with our construction.
\end{corollary}

\section{Permutation pattern matching}

Permutation pattern matching~\cite{bose-ppm} is an NP-complete problem, that asks given a permutation $\pi$ of $\{1, \dots, n\}$ and $\sigma$ of $\{1, \dots, \ell\}$ with $\ell \leq n$ if $\pi$ has a subsequence that are in the same relative order as $\sigma$. If such a subsequence exists, we say that $\pi$ \emph{matches} $\sigma$, otherwise, $\pi$ \emph{avoids} $\sigma$.

Let $x\in\B^{nk}$ be variables, and let $y\in\B^{\ell k}$ be a fixed permutation of $\{1,\dots,\ell\}$ written on $k=\lfloor\log(n+1)\rfloor$ bits.

\begin{proposition}[Permutation patterns]
Given $x \in \B^{nk}$ and $y \in \B^{\ell k}$, with $\ell \leq n$, there is a quadratic polynomial representation for the relation $\match$, that contains pairs $(x,y)$, where the $x$ and $y$ are both permutations and $x$ matches $y$. This representation uses $O(n \log^2 n)$ variables.
\end{proposition}

\begin{proof}
Enforce that the $x$ form a permutation by $\perm(x)$.

Denote the first $t$ variables of of $x$ by $x_{1..t} = (x_1, \dots, x_t)$. The constraint that $x_{1..\ell}$ matches $y$ is represented by the penalty
\begin{equation}
\label{eq-pattern-begin}
\nw(x_{1..\ell}, z, c) + \nw(y, \id, c)
\end{equation}
where the $z\in \B^{\ell k}$ are auxiliary variables and the $c \in \B^m$ are the control bits operating both sorting networks. This penalty is $0$ if and only if the $x_{1..\ell}$ are in the same relative order as the $y$. This constraint also enforces that the $y$ form a permutation.

We introduce bits $p=(p_1,\dots,p_n)\in\B^n$ to allow selecting any positions of $x$ to match $y$. Let $p_i=0$ if the position is selected and $p_i=1$ otherwise. To enforce that exactly $\ell$ positions are selected use the penalty $H_{n-\ell}(p)$. To move the selected $x_i$ to the first $\ell$ positions, while keeping them in the same relative order, we use a stable sorting network 
\begin{equation}
\label{eq-pattern-general}
\snw((p, \id; x), (q; w),
\end{equation}
where the $q \in \B^{n(k+1)}$ are auxiliary variables and the $w \in \B^{nk}$ contain the rearranged $x$.

Using (\ref{eq-pattern-begin}) on the $w$ results in the required constraint. This construction used four sorting networks with $O(k)$ bit width on each, which means this formulation uses $O(n \log^2 n)$ variables.
\end{proof}
Note that since multiple subsequences of $x$ may match $y$, this representation is not uniform.

\begin{credits}
\subsubsection{\ackname}
We thank András Bodor (HUN-REN Wigner RCP) for useful discussions.

Author Viktória Nemkin was supported by the Doctoral Excellence Fellowship Programme (DCEP), funded by the National Research, Development and Innovation Fund of the Ministry of Culture and Innovation and the Budapest University of Technology and Economics.

This paper was supported by the Ministry of Culture and Innovation and the National Research, Development and Innovation Office within the Quantum Information National Laboratory of Hungary (Grant No. 2022-2.1.1-NL-2022-00004).

\subsubsection{\discintname}
The authors have no competing interests to declare that are
relevant to the content of this article.
\end{credits}

\bibliographystyle{splncs04}
\bibliography{ce-qubo}

\newpage

\appendix
\section*{Appendix}

\section{Proof of Theorem \ref{thm-const-struct} and comparison with the permutation matrix representation}
\label{appx-const-struct}

\begin{proof}
(a) To enforce $\pi(i) = j$, replace the variable $x_i$ in the penalty $\perm(x)$ with the binary representation of the number $j$.

(b) To enforce $\pi(i) \neq j$, write $j$ in binary as $j^k, \dots, j^1$ and apply the penalty $\tresh_1$ from (\ref{eq-treshold}) on the $k$-bit number whose $b$th bit is
\begin{equation*}
x_i^b + j^b - j^b x^b_i.
\end{equation*}

(c) To enforce $i$ be a fixed point of $\pi$, i.e. $\pi(i) = i$, use (a).

(d) To enforce $\pi$ to be a derangement, i.e. have no fixed points, use (b) on all indices.

(e) To enforce $\pi \neq \tau$ for a fixed permutation $\tau$, write the corresponding $y_i$ in binary as $y_i^k, \dots, y_i^1$ and apply the penalty $\tresh_1$ from (\ref{eq-treshold}) on the $nk$-bit string whose $(i,b)$th bit is 
\begin{equation*}
x_i^b + y_i^b - y_i^b x_i^b.
\end{equation*}

All of these representations are uniform. Cases (a) and (c) require no additional auxiliary variables, while cases (b), (d) and (e) introduce auxiliaries to enforce the $T_1$ constraint, however these are negligible compared to the number of auxiliaries used in the construction of $\perm$.
\end{proof}

In the permutation matrix representation these constraints can be enforced in a similar fashion, by considering the rows of the matrix as unary representations of the $x_i$.

\section{Proof of Theorem \ref{thm-const-gt}  and comparison with the permutation matrix representation}
\label{appx-const-gt}

\begin{proof}
The following constraints can be enforced by adding further penalties to $\perm(x)$, while the number of auxiliary variables remains $O(n \log^2 n)$.

(a) To enforce $\sigma = \pi \pi'$, use the penalty
\begin{equation*}
\nw((x; y), (\id; x')) + \nw(x', \id).
\end{equation*}

(b) To enforce that $\pi$ is an involution, i.e., $\pi^2 = I$, which is a special case of (a), use the penalty
\begin{equation*}
\nw((x; \id), (\id; x)).
\end{equation*}

(c) To enforce that $\pi$ and $\pi'$ commute, i.e., $\pi' \pi = \pi \pi'$ use the penalty
\begin{equation*}
\nw((x; y), (\id; x')) + \nw((x'; y), (\id; x)).
\end{equation*}

(d) To enforce that $\pi'$ is a conjugate of $\pi$, i.e. $\pi' = \sigma^{-1} \pi \sigma$, use the penalty
\begin{equation*}
\nw((y, z), (\id; x)) + \nw((x', z), (\id; y)).
\end{equation*}

(e) To enforce the parity of $\pi$ to be even or odd, use the penalties
\begin{align*}
\nw(x, \id, c) + \even(c)\\
\nw(x, \id, c) + \odd(c).
\end{align*}

(f) To enforce $\pi^r = I$ for some constant $r \geq 3$ introduce variables $x_2, \dots, x_{r-1} \in \B^{nk}$ and use the penalty
\begin{equation*}
\nw((x; x_2, \dots, x_{r-1}, \id), (\id; x, x_2, \dots, x_{r-1})).
\end{equation*}

(g) To enforce that the order of $\pi$ is $r$, in addition to (f), use Theorem \ref{thm-const-struct} (e) with the identity permutation $\id$ on the $x_2, \dots, x_{r-1}$ variables.
\end{proof}

All of these representations are uniform and use $O(n \log^2 n)$ variables, since the dominant part of the formulation is the representation of $\nw$ which is used a constant number of times.

For the constraints corresponding to (a) in the permutation matrix representation, where permutations $\sigma, \pi, \pi'$ are encoded by matrices $S, P, P' \in \B^{n\times n}$ the matrix product $S = P' P$ is described by
\begin{equation*}
s_{i,j} = \sum_{k} p'_{i,k} p_{k,j} \quad (1 \leq i,j \leq n),
\end{equation*}
for which an additional penalty function similarly to (\ref{eq-hamming}) is
\begin{equation*}
\sum_{i=1}^n\sum_{j=1}^n (s_{i,j} - \sum_{k} p'_{i,k} p_{k,j})^2,
\end{equation*}
which can be quadratized by introducing auxiliary variables $z_{i,k,j} = p'_{i,k} p_{k,j}$, for all $1 \leq i,j,k \leq n$, for a total of $O(n^3)$ new variables.

The constraint in (b) can be enforced in the permutation matrix encoding by requiring the matrix to be symmetric, which introduces no additional variables. This approach, however, does not extend to (f) and (g), which generalize (b).

To enforce the constraints (c), (d), (f) and (g), repeated matrix multiplications can be done, similarly to (a), leading to a formulation with $O(n^3)$ new variables.

For the constraint (e) in permutation matrix encoding, for matrix $P \in B^{n\times n}$ the number of inversions is calculated by
\begin{equation*}
\sum_{i<j, k>\ell} x_{i,k} x_{j,\ell},
\end{equation*}
for which an additional penalty function similarly to (\ref{eq-even}) and (\ref{eq-odd}) can be introduced. To obtain a quadratic penalty, the sum above must be linearized, by introducing auxiliary variables for each $x_{i,k} x_{j,\ell}$ product, for all $1 \leq i,j,k,\ell \leq n$, for a total of $O(n^4)$ additional variables.

\end{document}